\documentclass[a4paper,11pt]{article}

\usepackage[utf8]{inputenc}
\usepackage{amsthm,amsmath,amssymb,amsfonts,tikz}
\usepackage{color}
\usepackage[bookmarksnumbered=true]{hyperref}

\usepackage{titlesec}
\titleformat*{\section}{\Large \normalfont \bfseries}
\titleformat*{\subsection}{\large \normalfont \bfseries}
\usepackage{setspace}
\usetikzlibrary{arrows,positioning,cd} 

\newcommand{\nc}{\newcommand}
\nc{\be}{\begin{equation}} \nc{\ee}{\end{equation}}
\nc{\ba}{\begin{array}} \nc{\ea}{\end{array}}
\nc{\bea}{\begin{eqnarray}} \nc{\eea}{\end{eqnarray}}
\nc{\ny}{\nonumber}
\nc{\ra}{\rangle}
\nc{\la}{\langle}
\nc{\lb}{\left(}
\nc{\rb}{\right)}
\nc{\pt}{\partial}
\nc{\D}{\Delta}
\nc{\kr}{\mathcal}
\nc{\lmb}{\lambda}
\nc{\sg}{\sigma}
\nc{\bbl[1]}{\lbrack #1\rbrack{}}
\nc{\ep}{\epsilon}
\nc{\ora}{\overrightarrow}
\nc{\td}{\widetilde}
\nc{\tq}{\theta}
\nc{\seteq}{\mathbin{:=}}
\nc{\simto}{\xrightarrow{\,\sim\,}}
\nc{\Z}{{\mathbb Z}}

\nc{\nsr}{\scriptscriptstyle{\textsf{NSR}}}

\nc{\Vir}{\textsf{Vir}}
\nc{\Asl}{\widehat{\mathfrak{sl}}}
\nc{\rmF} {\mathrm{F}}
\nc{\calM}{\mathcal{M}}
\nc{\calA}{\mathcal{A}}
\nc{\calF}{\mathcal{F}}
\nc{\calL}{\mathcal{L}}
\nc{\calU}{\mathcal{U}}
\nc{\calE}{\mathcal{E}}
\nc{\calV}{\mathcal{V}}
\nc{\calH}{\mathcal{H}}
\nc{\sfV}{\textsf{V}}
\nc{\tr}{\mathrm{Tr}}
\nc{\red}{\color{red}}
\nc{\NS}{\scriptscriptstyle{\textsf{NS}}}
\nc{\vac}{\varnothing}
\nc{\G}{\textsf{G}}
\nc{\pure}{\mathrm{pure}}
\nc{\ul}{\underline}
\nc{\ol}{\overline}
\numberwithin{equation}{section}
\newtheorem{thm}{Theorem}[section]
\newtheorem{prop}{Proposition}[section]

\newtheorem{Remark}{Remark}[section]
\newtheorem{conj}{Conjecture}[section]

\newtheorem{defn}{Definition}[section]
\newtheorem{example}{Example}[section]
\tikzset{
    >=stealth',
    pil/.style={
           ->,
           shorten <=0pt,
           shorten >=0pt,},
        }

\title{$q$-deformed Painlev\'e $\tau$ function and $q$-deformed conformal blocks}
\date{}
\author{M.~A.~Bershtein, A.~I.~Shchechkin}
\begin{document}
\maketitle 
\begin{abstract}\vspace*{2pt}
We propose $q$-deformation of the Gamayun-Iorgov-Lisovyy formula for Painlev\'e $\tau$ function. 
Namely we propose the formula for $\tau$ function for $q$-difference Painlev\'e equation corresponding
to $A_7^{(1)\prime}$ surface (and $A_1^{(1)}$ symmetry) in Sakai classification.
In this formula $\tau$ function equals the series of $q$-Virasoro Whittaker conformal blocks
(equivalently Nekrasov partition functions for pure $SU(2)$ 5d theory).

\end{abstract}

\tableofcontents

\newpage

\section{Introduction}
\label{sec:intro}
The goal of this paper is to find $q$-deformation of suggested in \cite{GIL1207},\cite{GIL1302} formulas for Painlev\'e $\tau$ functions. More precisely we will do this for case of Painlev\'e III($D_8$). 

We briefly recall necessary details about this equation. In standard form it is a nonlinear second order differential equation on function $w(z)$, namely
\begin{equation}
w''=\frac{w'^2}{w}-\frac{w'}{z}+\frac{2w^2}{z^2}-\frac{2}{z}
\label{Pp} 
\end{equation}
This equation can be rewritten as system of two bilinear Toda-like equations on two $\tau$ functions \cite{BSRamon} 
\begin{equation}
1/2 D^2_{[\log z]}(\tau(z),\tau(z))=-z^{1/2}\tau_1(z)\tau_1(z),\quad 1/2 D^2_{[\log z]}(\tau_1(z),\tau_1(z))=-z^{1/2}\tau(z)\tau(z),
\end{equation}
where $D^2_{[\log z]}$ denotes second Hirota operator with respect to $\log z$, and $\tau_1$ is a B\"acklund transformation of  $\tau$ (the group of B\"acklund transformations of this equation is $\mathbb{Z}_2$). The function $w(z)$ is equal to $z^{1/2}\tau(z)^2/\tau_1(z)^2$, B\"acklund transformation acts on $w$ as $w \mapsto z/w$.

The Gamayun--Iorgov--Lisovyy formula for the $\tau$ function has the form \cite{GIL1302}
\begin{equation}
\tau(\sigma,s|z)=\sum_{n\in\mathbb{Z}}C(\sigma+n)s^n z^{(\sg+n)^2} \mathcal{F}((\sigma+n)^2|z), \label{GIL}
\end{equation}
where $s,\sigma$ are integration constants.  The function $\mathcal{F}(\Delta|z)$ denotes Whittaker limit of Virasoro conformal block in a representation of the highest weight $\Delta$ and the central charge $c=1$. The function $C(\sigma)=1/(\G(1-2\sigma)\G(1+2\sigma))$, where $\mathsf{G}$ is a Barnes $\mathsf{G}$ function. Function $\tau(\sigma,s|z)$ satisfies  
evident property $\tau(\sigma,s|z)=s^{-1}\tau(\sigma+1,s|z)$ (periodicity in $\sigma$) and B\"acklund transformation acts as $\tau_1 \propto \tau(\sigma+1/2,s|z)$. Therefore bilinear equations and formula for $w$ can be rewritten as (see \cite{BSRamon} for details)
\begin{equation}\label{sgToda} 
1/2 D^2_{[\log z]}(\tau(\sigma,s|z),\tau(\sigma,s|z))=-z^{1/2}\tau(\sigma+1/2,s|z)\tau(\sigma-1/2,s|z),
\end{equation}
\begin{equation}
w(z)=z^{1/2}\frac{\tau(\sigma,s|z)^2}{\tau(\sigma-1/2,s|z)\tau(\sigma+1/2,s|z)}\label{tautau1}
\end{equation}
The formula \eqref{GIL} was proven in \cite{ILT} and \cite{KGP}. Remark that the proof in \cite{KGP} is based on bilinear relations on Virasoro conformal blocks. 

It is natural to expect that there exists certain $q$-deformation of formula \eqref{GIL} which gives $\tau$~function for $q$-difference Painlev\'e equation. We follow Sakai approach to these equations \cite{SakaiCMP}, where discrete equations are associated with certain rational surfaces which are obtained by blow up of 9 points in $\mathbb{CP}^2$. These surfaces are parametrized by two complementary sublattices in Picard group of the surface, the last lattice is isomorphic to $\mathbb{Z}^{1,9}$. These two sublattices are denoted as $R$ and $R^\perp$ and specify so-called surface type and symmetry type correspondingly. The celebrated Sakai tables are given in Figs. \ref{Fig:Sakaisurf} and \ref{Fig:Sakaisym}.
\begin{figure}[h]
\begin{center}
	{\small
\begin{tikzcd}[row sep=scriptsize, column sep=scriptsize]
	& & & & & & & {A_7^{(1)}} \arrow[r d] &\\
	{A_0^{(1)}} \arrow[r] & {A_1^{(1)}}  \arrow[r] &	{A_2^{(1)}} \arrow[r] & {A_3^{(1)}}  \arrow[r] \arrow[r d] &	{A_4^{(1)}} \arrow[r] \arrow[d r] & {A_5^{(1)}} \arrow[r d] \arrow[r] \arrow[r d d] &	{A_6^{(1)}} \arrow[r] \arrow[r d] \arrow[r u] & {\boxed{A_7^{(1)\prime}}} \arrow[r d] \arrow[d d, bend left=30] & {A_8^{(1)}} \arrow[d d, bend left=30]
 \\
	& & & & D_4^{(1)}  \arrow[r]& D_5^{(1)} \arrow[r]  \arrow[rd]& D_6^{(1)} \arrow[r d] \arrow[r] & D_7^{(1)} \arrow[r] & D_8^{(1)}\\
	& & & &  &  & E_6^{(1)}\arrow[r] & E_7^{(1)}\arrow[r] & E_8^{(1)}
\end{tikzcd}}
\end{center}
\caption{Sakai classification, surface type \label{Fig:Sakaisurf}}
\end{figure}
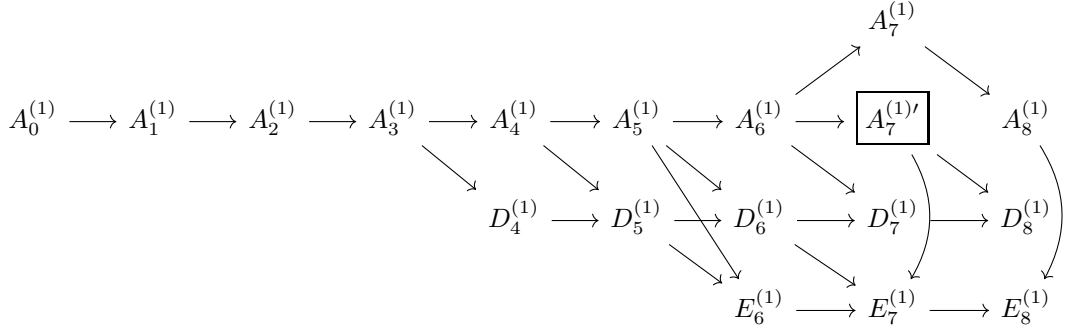
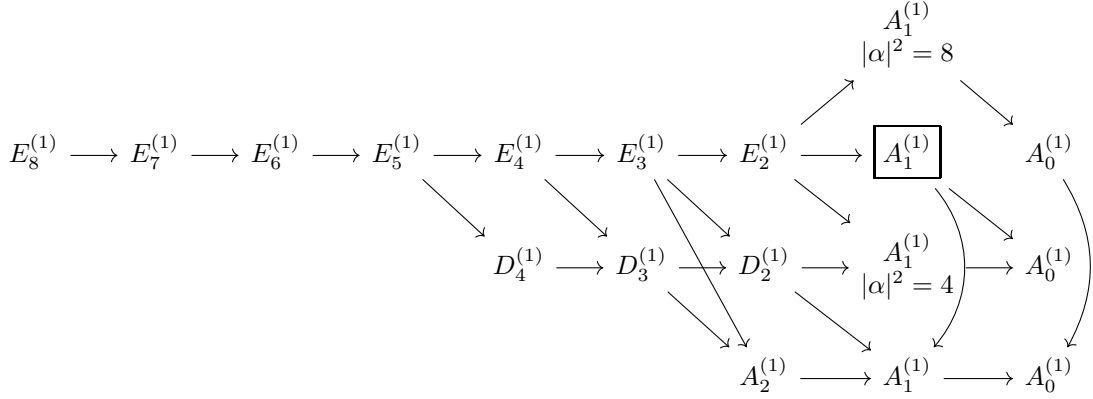
\begin{figure}[h]
\begin{center}
	{\small
	\begin{tikzcd}[row sep=scriptsize, column sep=scriptsize]
		& & & & & & & {\begin{matrix}A_1^{(1)} \\|\alpha|^2=8\end{matrix}} \arrow[r d] &\\
		{E_8^{(1)}} \arrow[r] & {E_7^{(1)}}  \arrow[r] &	{E_6^{(1)}} \arrow[r] & {E_5^{(1)}}  \arrow[r] \arrow[r d] &	{E_4^{(1)}} \arrow[r] \arrow[dr] & {E_3^{(1)}} \arrow[r d] \arrow[r] \arrow[r d d] &	{E_2^{(1)}} \arrow[r] \arrow[r d] \arrow[r u] & {\boxed{A_1^{(1)}}} \arrow[r d] \arrow[d d, bend left=40] & {A_0^{(1)}} \arrow[d d, bend left=30]\\
		& & & & D_4^{(1)}  \arrow[r]& D_3^{(1)} \arrow[r]  \arrow[r d]& D_2^{(1)} \arrow[r d] \arrow[r] & {\begin{matrix}A_1^{(1)} \\|\alpha|^2=4\end{matrix}} \arrow[r] & A_0^{(1)}\\
		& & & &  &  & A_2^{(1)} \arrow[r]& A_1^{(1)}\arrow[r] & A_0^{(1)}
	\end{tikzcd}}
	\end{center}
\caption{Sakai classification, symmetry type \label{Fig:Sakaisym}}
\end{figure}
Here $X_n^{(1)}$ are standard notations for affine root lattices. We do not explain precise meaning of notations like $D_2^{(1)}$ or $A_1^{(1)}, |\alpha|^2=4$, see \cite{SakaiCMP} or \cite{KNY} for details. Arrows in these figures indicate degenerations of the surfaces.

Recall that continuous Painlev\'e equations also correspond to the rational surfaces, namely for each Painlev\'e equation there exists so-called space of initial conditions \cite{Okamoto}. This space for \eqref{Pp} is the surface of $D_8^{(1)}$ type. This is the reason why this equation is called Painlev\'e III($D_8$). The $q$-deformation of this equation corresponds to the deformation of the surface (see \cite[Sec 7]{SakaiCMP}). Therefore we consider $q$-difference Painlev\'e equation corresponding to the surface type $A_7^{(1)\prime}$ and the symmetry type $A_1^{(1)}$ (the corresponding sublattices are drawn in box in Figs.~ \ref{Fig:Sakaisurf},~\ref{Fig:Sakaisym}). 

Such surface depends on two parameters which we denote by $q$ and $Z$ ($q$-deformed analogue of $z$ in \eqref{Pp}). Details are given Subsection \ref{ssec:Cremona}, here we write only $q$-deformation of the equation \eqref{Pp} 
\begin{equation}
G(qZ)G(q^{-1}Z)=\left(\frac{G(Z)-Z}{G(Z)-1}\right)^2
\end{equation}
This equation as well as its relation to Painlev\'e III$(D_8)$ is given in \cite{Ramani:2015}. We did not found in the literature the definition of $\tau$ function for this difference equation. So we propose the definition following Tsuda approach \cite{Tsuda2006} (see also \cite{KNY}). It is convenient to use four $\tau$ functions $\mathcal{T}_1, \mathcal{T}_2, \mathcal{T}_3, \mathcal{T}_4$. This result is given in Theorem \ref{Th:tauformal}. 

The next task is to give a formula for $\mathcal{T}_i$ as an analytical function on variable $Z \in \mathbb{C}$ (above one can think that $\mathcal{T}_i$ is defined on the lattice as in \cite{Tsuda2006}). As in continuous case the conditions on the $\tau$ functions  can be rewritten as a  system of equations on one function $\mathcal{T}(u,s;q|Z)$
\begin{equation}
\begin{aligned}
Z^{1/4}\mathcal{T}(u,s;q|qZ)\mathcal{T}(u,s;q|q^{-1}Z)&=\mathcal{T}(u,s;q|Z)^2+Z^{1/2}\mathcal{T}(uq,s;q|Z)\mathcal{T}(uq^{-1},s;q|Z)\\
\mathcal{T}(uq^2,s;q|Z)=s^{-1}\mathcal{T}(u,s;q|Z) . 
\end{aligned} \label{eq:desTau}
\end{equation}
Here $u, s$ are $q$-deformed analogs of parameters $\sigma,s$ in \eqref{GIL}.

In Section \ref{sec:deform} we propose the solution of this system in the ansatz
\begin{equation}
\mathcal{T}(u,s;q|Z)=\sum_{n \in \mathbb{Z}}s^n C(uq^{2n};q|Z) \frac{\mathcal{F}(uq^{2n};q^{-1},q|Z)}{(uq^{2n+1};q,q)_\infty(u^{-1}q^{-2n+1};q,q)_\infty}
\label{eq:T}
\end{equation} 
Here $\mathcal{F}(u;q^{-1},q|Z)$ denotes Whittaker limit of the $q$-deformed conformal block,
$(u;q,q)_\infty$ denotes double infinite Pochhammer symbol and the function $C(u;q|Z)$ is defined by its second difference derivatives, see Definition  \ref{Def:tau}. The function  $C(u;q|Z)$ is $q$-deformed analogue of $z^{\sigma^2}$ in formula \eqref{Pp}.
We give a couple of examples of such functions in Example \ref{Ex:C}.

In Subsection \ref{ssec:qdeftau} we discuss validity of \eqref{eq:desTau} for function
$\mathcal{T}(z)$ defined in \eqref{eq:T}. The second condition in \eqref{eq:desTau} is straightforward from \eqref{eq:T}. The bilinear relation in \eqref{eq:desTau} is equivalent to the bilinear
relation on function $\mathcal{F}(u;q^{-1},q|Z)$, see \eqref{bilconfrel}
and Theorem \ref{mainbilin}. We do not prove corresponding bilinear relation but give several arguments in support.

In order to finish description of the paper content note
that in Subsections \ref{ssec:contlimeq} and	 \ref{ssec:contlim} we discuss $q \rightarrow 1$ limit, in Subsection \ref{ssec:algsol} we discuss special values of $u,s$
which correspond to $q$-deformation of algebraic solution of \eqref{Pp}. In Appendix \ref{app:q} we collect necessary definitions and properties
of $q$-special functions. In Appendix \ref{app:toralg} we write more general bilinear relations on $q$-deformed conformal blocks.

\section{$q$-deformed equations}
\label{sec:repr}
\subsection{$A_7^{(1)\prime}$ surface and its symmetry group}
\label{ssec:Cremona}
Here we follow Sakai approach to the $q$-difference Painlev\'e equations associated with rational surfaces \cite{SakaiCMP}, \cite{KNY}.
In this approach $q$-deformed Painlev\'e equation arises from discrete group of birational automorphisms acting on a rational surface.

Let $X$ be a rational surface of the type $A_7^{(1)\prime}/A_1^{(1)}$ in Sakai notations \cite{SakaiCMP} where first notation denotes surface type 
and second notation denotes symmetry type of the surface.
It is obtained by blow-up of $\mathbb{CP}^2$ in nine points. Scheme of this procedure one can see in Fig.~\ref{blowup}
(see \cite[App. B, Mul.~9]{SakaiCMP}). 

\begin{figure}[h]
\begin{center}
\begin{tikzpicture}[x=2em,y=2em]
\draw (-1,-0.5) -- (0.5, 2) node[anchor=west] {$x=0$};
\draw (-1.4, -0.1) -- (2,-0.1) node[anchor=west] {$y=0$};
\draw (1.6,-0.5) -- (0.1,2) node[anchor=east] {$z=0$};
\path[fill=black] (-0.25,0.75) circle (0.7mm) node[anchor=east] {$p_4$};
\path[draw=black] (0.3,1.67) circle (1.5mm);
\draw (0.5,1.8) node[below right] {$p_7\leftarrow p_8 \leftarrow p_9$};
\path[fill=black] (1.36,-0.1) circle (0.7mm);
\path[draw=black] (1.36,-0.1) circle (1.5mm);
\draw (1.36,0) node[above right] {$p_2\leftarrow p_6$};
\path[fill=black] (-0.76,-0.1) circle (0.7mm);
\path[draw=black] (-0.76,-0.1) circle (1.5mm);
\draw (-1,-0.4) node[anchor=east] {$p_5\rightarrow p_3\rightarrow p_1$};
\end{tikzpicture}
\caption{\leftskip=1in \rightskip=1in Blow-up scheme for $X$}
	\label{blowup}
\end{center}
\end{figure}
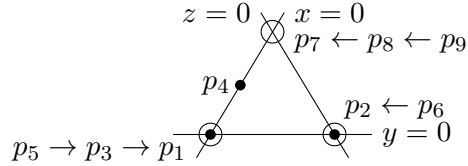

Let $(x:y:z)$ denote coordinates of $\mathbb{CP}^2$, $(a,b,c)$ are parameters of surface which characterize positions of points of blow-up
(our $a$ is $a_1$ in \cite{SakaiCMP}). There is an equivalence $(a,b,c;x:y:z)\sim (a,\mu^2 b,\mu^{-2}c;\mu x:y:z)$.
There exists discrete subgroup of Cremona group (group of rational transformations of $\mathbb{CP}^2$) which preserve blow-up structure.
We will denote it by $W$. Group $W$ is equal to $Dih_4\ltimes W(A^{(1)}_1)$ where $Dih_4$ is dihedral group of square and $W(A^{(1)}_1)$ is Weyl group of
$\widehat{sl}(2)$. Group $W$ can be presented by generators $s_0, s_1, \pi_1, \pi_2$ and relations as
\begin{equation} \label{eq:Crem:rel}
s_0^2=s_1^2=1,\quad \pi_1^2=\pi_2^4=(\pi_1\pi_2)^2=1,\quad s_1=\pi_2s_0\pi_2^{-1},\quad s_0=\pi_2^2s_0\pi_2^{-2}=\pi_1s_0\pi_1^{-1},
\end{equation}
where $s_0,s_1$ are Weyl group simple reflections, $\pi_1, \pi_2$ are generators of $Dih_4$.
Structure of $W$ could be graphically represented on Fig. \ref{dihedral}. First two pictures define action of the generators of dihedral group on the square
and the last picture describes structure of the semidirect product, namely $\pi_2$ acts as the outer automorphism which interchange roots.

\begin{figure}[h]
\begin{center}
\begin{tikzpicture}

\draw (6,1) node [anchor=east] {$\alpha_0$}
edge[pil,bend right=0] (5.999,1);
\draw (8,1) node [anchor=west] {$\alpha_1$}
edge[pil,bend right=0] (8.001,1); 
\draw (6.1,1.02) -- (7.9,1.02);
\draw (6.1,0.98) -- (7.9,0.98);
\draw (6.2,1.3)
edge[pil,<->,bend left=20] (7.8,1.3);
\draw (7,1.7) node {$\pi_2$};

\draw (-4,0) node [below left]  {$4$} -- 
        (-2,0) node [below right] {$3$} -- 
        (-2,2) node [above right] {$2$} -- 
        (-4,2) node [above left]  {$1$} -- 
        cycle;
        \draw[dashed] (-4.4, 2.4) -- (-1.2,-0.8) node[below right] {$\pi_1$};
        \draw (-1.7,-0.8)
        edge[pil,<-> ,bend right=45] (-1.25,-0.2);
        

        \draw (2,0) node [below left]  {$4$} -- 
        (4,0) node [below right] {$3$} -- 
        (4,2) node [above right] {$2$} -- 
        (2,2) node [above left]  {$1$} -- 
        cycle;
        
        \draw (3,1) node {$\pi_2$};
        \draw (4.3,0.2)
        edge[pil, bend right=20] (4.3,1.8);

        \draw (3.8,2.3)
        edge[pil, bend right=20] (2.2,2.3);

        \draw (1.7,1.8)
        edge[pil, bend right=20] (1.7,0.2);

        \draw (2.2,-0.3)
        edge[pil, bend right=20] (3.8,-0.3);
\end{tikzpicture}
\caption{\leftskip=1in \rightskip=1in Action of $Dih_4$ generators}
	\label{dihedral}
\end{center}
\end{figure}
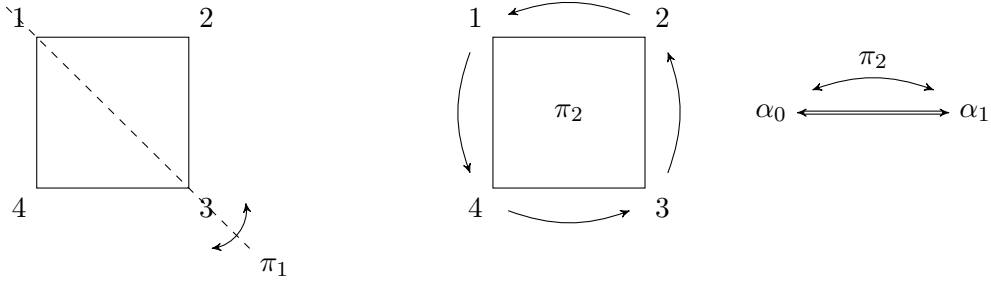

As a subgroup of Cremona group $W$ acts on rational surfaces by rational transformations\footnote{there is a misprint in formula for $s_0$
in \cite[App. C, Mul 9]{SakaiCMP}}
\begin{equation}
\begin{aligned}
\pi_1&\colon (a,b,c;x:y:z) \mapsto (1/a,1/c,1/b;x:z:y)
\\ 
\pi_2&\colon (a,b,c;x:y:z) \mapsto (bc,a/b,b;xy:byz:x^2)\\\
s_1&\colon (a,b,c;x:y:z) \mapsto (1/a,ab,ac;x(y-az):y(z-y):az(z-y))\\
s_0&\colon (a,b,c;x:y:z) \mapsto (ab^2c^2,1/c,1/b;x(x^2/b-yz)(yz-cx^2):y(yz-cx^2)^2:z(x^2/b-yz)^2)\label{Cract}
\end{aligned}
\end{equation}
 
Let us denote $q=abc$, $Z=a^{-1}$ and $F=yzb/x^2$, $G=z/y$. These coordinates are unambiguously defined (taking into account equivalence of $\mu$
rescaling). Then, in terms of $(Z,q,F,G)$ action of group $W$ has the form:
\begin{equation}\label{eq:WZqFG}
\begin{aligned}
\pi_1&\colon (Z,q,F,G) \mapsto (Z^{-1},q^{-1},q^{-1}Z^{-1}F,G^{-1}),
\\ 
\pi_2&\colon (Z,q,F,G) \mapsto (Z^{-1}q^{-1},q,Z^{-1} G,F^{-1}),\\\
s_1&\colon (Z,q,F,G) \mapsto (Z^{-1},q,F\frac{(G-1)^2}{(G-Z)^2},Z^{-1}G),\\
s_0&\colon (Z,q,F,G) \mapsto (Z^{-1}q^{-2},q,q^{-1}Z^{-1} F,G\frac{(1-F)^2}{(Z q-F)^2}).
\end{aligned}
\end{equation}
Introduce element $T=\pi_2^{-1}\circ s_0$ of infinite order in $W$. We will denote $\overline{x}=T(x)$ and $\underline{x}=T^{-1}(x)$.
Note that group $W$ is generated by $T,s_1$ and $\pi_1$.
We have 
\begin{equation}
(\overline{Z},\overline{q},\overline{F},\overline{G})=\left(q Z,q,\frac{(F-q Z)^2}{(F-1)^2 G},F\right),\qquad (\underline{Z},\underline{q},
\underline{F},\underline{G})=\left(q^{-1}Z,q,G,\frac{(G-Z)^2}{F (G-1)^2}\right).\ny
\end{equation}
Therefore
\begin{equation}
\overline{G}\underline{G}=\left(\frac{G-Z}{G-1}\right)^2 \label{qPp}.
\end{equation}
This equation could be called $q$-deformed Painlev\'e III($D_8$).
That is because if we consider $G$ as a function on $Z$ (and $\overline{G}=G(qZ), \, \underline{G}=G(q^{-1}Z)$) then
continuous limit of \eqref{qPp} is Painlev\'e III($D_8$) equation \eqref{Pp} (see Subsection \ref{ssec:contlimeq}). Of course this equation is not new, here we get it from the known formulas for $A_7^{(1)\prime}$ surface, but exactly in form \eqref{qPp} it was written in \cite[eq. (20)]{Ramani:2015}, see also earlier work \cite{Ramani:2000}.

\subsection{$\tau$ function representation of group $W$}
\label{ssec:repr}
We want to lift representation of group $W$ to the level of the $4$ letters $\mathcal{T}_i,\, i=1,\ldots 4$.
Analogously to \cite{Tsuda2006}(see also \cite{KNY}) we call $\mathcal{T}_i$ $\tau$ functions and state Theorem

\begin{thm}\label{Th:tauformal}
Action of generators $s_1, \pi_1, \pi_2$ of group $W$ on $\mathcal{T}_i,\, i=\overline{1,4}$ given by Table \ref{lettertable} 
provides a representation of $W$ in the field $\mathbb{C}(\mathcal{T}_1, \mathcal{T}_2, \mathcal{T}_3, \mathcal{T}_4, q^{1/4}, Z^{1/4})$.
\end{thm}
\begin{table}
\scalebox{0.95}{\begin{tabular}{|p{0.2cm}|c|p{0.4cm}|c|c|c|c|c|c|}
\hline
&$Z$ & $q$&$F$ & $G$ &$\mathcal{T}_1$& $\mathcal{T}_2$&$\mathcal{T}_3$&$\mathcal{T}_4$\\
\hline
$\pi_1$&$1/Z$ & $1/q$&$F/(qZ)$&$1/G$ &$\mathcal{T}_3$& $\mathcal{T}_2$&$\mathcal{T}_1$&$\mathcal{T}_4$\\
\hline
$\pi_2$& $1/(qZ)$&$q$&$ G/Z$&$1/F$&$\mathcal{T}_4$&$\mathcal{T}_1$&$\mathcal{T}_2$&$\mathcal{T}_3$\\
\hline
$s_1$&$1/Z$&$q$&$F\frac{(G-1)^2}{(G-Z)^2}$&$G/Z$&$\mathcal{T}_1$&$\dfrac{\mathcal{T}_3^2{+}Z^{1/2}\mathcal{T}_1^2 }{Z^{1/4}\mathcal{T}_4 }$&$\mathcal{T}_3$&$\dfrac{\mathcal{T}_1^2{+}Z^{1/2}\mathcal{T}_3^2 }{Z^{1/4}\mathcal{T}_2}$\\[10pt]
\hline
$s_0$&$1/(q^2Z)$&$q$&$F/(qZ)$&$G\frac{(1-F)^2}{(Z q-F)^2}$&$\dfrac{\mathcal{T}_4^2{+}(qZ)^{1/2}\mathcal{T}_2^2}{ (qZ)^{1/4} \mathcal{T}_3}$&$\mathcal{T}_2$&$\dfrac{\mathcal{T}_2^2{+}(qZ)^{1/2}\mathcal{T}_4^2}{ (qZ)^{1/4} \mathcal{T}_1}$&$\mathcal{T}_4$\\[10pt]
\hline
$T$&$qZ$&$q$&$\frac{(F-q Z)^2}{(-1+F)^2 G}$&$F$&$\mathcal{T}_2$&$\dfrac{\mathcal{T}_2^2{+}(q Z)^{1/2}\mathcal{T}_4^2}{(q Z)^{1/4}\mathcal{T}_1 }$&$\mathcal{T}_4$&$\dfrac{\mathcal{T}_4^2{+}(q Z)^{1/2}\mathcal{T}_2^2 }{(q Z)^{1/4}\mathcal{T}_3}$\\[10pt]
\hline
\end{tabular}}
\caption{Representation of $W$ on $\mathcal{T}_i$}
\label{lettertable}
\end{table}
Proof is straightforward. Generators $s_1, \pi_1, \pi_2$ generate the whole group. So it is enough to check relations of group
$s_1^2=1, \pi_1^2=\pi_2^4=(\pi_1\pi_2)^2=1, s_1=\pi_1s_1\pi_1^{-1}=\pi_2^2s_1\pi_2^{-2}$ for given actions.

Subgroup $Dih_4$ of $W$ acts on $\tau$ functions the same as it acts on square vertices on Fig. \ref{dihedral} i.e.  $w(\mathcal{T}_i)=\mathcal{T}_{w(i)}$ for any $w \in Dih_4$. In more geometrical language one can assign 
$\tau$ functions to the four blow-up points in $\mathbb{CP}^2$ (cf. \cite{Tsuda2006}). 

Define $F,G$ by the formulas
\begin{equation}
F=-(qZ)^{1/2}\frac{\mathcal{T}_2^2}{\mathcal{T}_4^2}, \qquad G=-Z^{1/2}\frac{\mathcal{T}_1^2}{\mathcal{T}_3^2}.\label{FGdefinition}
\end{equation}
This is just an discrete analog of \eqref{tautau1}.
Then one can check that action of $W$ on $F$ and $G$ induced by action on $\tau$ functions coincides with the action defined in  \eqref{eq:WZqFG}.

The formula for action $T$ implies that $\mathcal{T}_2=\overline{\mathcal{T}_1}$ and $\mathcal{T}_4=\overline{\mathcal{T}_3}$. Then from the action of $T$ of $\mathcal{T}_2$ and $\mathcal{T}_4$ we get
\begin{equation}\label{eq:Tau13}
Z^{1/4}\overline{\mathcal{T}_1}\underline{\mathcal{T}_1}=\mathcal{T}_1^2+Z^{1/2}\mathcal{T}_3^2,\qquad Z^{1/4}\overline{\mathcal{T}_3}\underline{\mathcal{T}_3}=\mathcal{T}_3^2+Z^{1/2}\mathcal{T}_1^2.
\end{equation}

\begin{Remark}
   It is interesting to note that under the action of group $W$ maps $\mathcal{T}_i$ to the Laurent polynomials on $\mathcal{T}_i$ non just rational functions. This property is nontrivial, for example  this is not true for $G$ and $F$. This observation follows from the fact that action of $W$ can be represented as composition of mutations for cluster algebra\footnote{we are grateful to P. Gavrylenko for the discussion of this point} similar to the higher $q$-deformed Painlev\'e equation in the paper~\cite{Okubo}. 	
\end{Remark}

\begin{Remark} Equations corresponding the $A_7^{(1)\prime}$ surface are usually written in bit more general form then \eqref{qPp}. For example in \cite[eq. (10)]{Ramani:2015}, (see also \cite[eq.(2.19))]{Grammaticos:2002}): 
\begin{equation}\label{eq:A7'GR}
(\overline{x}y-1)(xy-1)=Zy^2,\quad(xy-1)(x\underline{y}-1)=Z.		
\end{equation}
This system can be solved in terms of $\tau$ functions subject of \eqref{eq:Tau13}. Namely 
\begin{equation*}
x=q^{-1/4}Z^{1/4}\frac{\underline{\mathcal{T}_3}\overline{\mathcal{T}_1}}{\mathcal{T}_3\mathcal{T}_1},\quad y=q^{1/4}\frac{\overline{\mathcal{T}_3}{\mathcal{T}_1}}{\mathcal{T}_3\overline{\mathcal{T}_1}}.
\end{equation*}
Therefore one can consider system \eqref{eq:Tau13} as a bilinear form of \eqref{eq:A7'GR}. In the paper \cite[eq. (2.23)-(2.24)]{SakaiLax} this system is written in a different form 
\begin{equation}\label{eq:A7'Sakai}
f\overline{f}={g(Z-g)}/{(g-1)},\quad
g\overline{g}=\overline{f}^2.
\end{equation}
This system can be also solved in in terms of $\tau$ functions and in fact equvalent to \eqref{eq:A7'GR} by the invertible transformation 
\begin{equation*}
 g=1-xy, \quad f=-q^{1/2}\underline{g}/\underline{y}.
\end{equation*}
\end{Remark}

We want to consider $F,G, \mathcal{T}_i$ as functions on $Z$ (and $q$), i.e. $\overline{H}=H(qZ)$, $\underline{H}=H(q^{-1}Z)$ where $H=F, G, \mathcal{T}_i$. The equations \eqref{eq:Tau13} become $q$-difference bilinear equations on $\mathcal{T}_1$ and $\mathcal{T}_3$. In the next section we give a solution of these equations in terms of one function $\mathcal{T}(u,s;q|Z)$ depending on two parameters $u,s$. Below we will use parameter $\sigma$, such that $u=q^{2\sigma}$, the parameters $\sigma, s$ are direct analogues of $\sigma,s$ for continuous Painlev\'e in the formula \eqref{GIL}. Geometrically $s,u$ can be viewed as parametrization of the (open subset of) surface $X$. We set 
\begin{equation}\label{eq:T1T3}
\mathcal{T}_1=\mathcal{T}(u,s;q|Z), \quad 
\mathcal{T}_3=s^{1/2}\mathcal{T}(uq,s;q|Z),
\end{equation}
and assume quasi periodicity in $u$
\begin{equation}
\mathcal{T}(uq^2,s;q|Z)=s^{-1}\mathcal{T}(u,s;q|Z).
\end{equation}
Then the equations \eqref{eq:Tau13} reduces to 
\begin{equation}
Z^{1/4}\mathcal{T}(u,s;q|qZ)\mathcal{T}(u,s;q|q^{-1}Z)=\mathcal{T}(u,s;q|Z)^2+Z^{1/2}\mathcal{T}(uq,s;q|Z)\mathcal{T}(uq^{-1},s;q|Z)\label{TodaT}.
\end{equation}

The functions $F(Z), G(Z)$ are defined in terms of function $\mathcal{T}$ by formula \eqref{FGdefinition}. Note that solutions of equations \eqref{eq:Tau13} depend on four parameters, in addition of $u, s$ we have simple symmetry $(\mathcal{T}_1,\mathcal{T}_3) \mapsto (\alpha Z^{\beta}\mathcal{T}_1,\alpha Z^{\beta}\mathcal{T}_3)$. This symmetry do not acts on the functions $F(Z), G(Z)$.

Remark that element $\pi_2^2 \in W$ commutes with shift $T$, in the continuous limit below this transformation goes to the B\"acklund transformation. It acts by formula $\pi_2^2 \colon \mathcal{T}(u,s;q|Z) \mapsto \mathcal{T}(uq,s;q|Z)$. In terms of parameter $\sigma$ this is transformation $\sigma \mapsto \sigma+1/2$ as for continuous Painlev\'e.

\subsection{Continuous limit of the equations}
\label{ssec:contlimeq}
Let us check that continuous limit of \eqref{qPp} is \eqref{Pp} as it should be 
for $q$-deformed Painlev\'e III($D_8$).
Introduce notations $q=e^{\hbar}, \, u=q^{2\sg}$.
In the limit questions we assume that $|q|<1$ $q \rightarrow 1$.

\begin{prop} \label{Geqlim}
Substitute
\begin{equation}
Z=\hbar^4 z, \quad
G(Z)=\hbar^2 w(z).
\end{equation}
If $\hbar\rightarrow 0$ the leading order of the equation \eqref{qPp} is Painlev\'e III($D_8$) \eqref{Pp} equation.
\end{prop}
Proof is by direct calculation, namely expansion of \eqref{qPp} into powers of $\hbar$. We obtain that first non-trivial coefficient with power $\hbar^2$. This coefficient vanishes and this is just Painlev\'e III($D_8$)~\eqref{Pp}.

We also check the continuous limit of bilinear equation \eqref{TodaT}.
It is convenient to define new function $\mathcal{T}_c(u,s;q|Z)$ by
\begin{equation}
\mathcal{T}_c(u,s;q|Z)=\frac{(q;q,q)^2_{\infty}}{\Gamma(-(qZ)^{1/4};q^{1/4},q^{1/4})}\mathcal{T}(u,s;q|Z)\label{redqToda}.
\end{equation}
Equation \eqref{TodaT} in terms of function $T_c(u,s;q|Z)$ reads
\begin{equation*}
\mathcal{T}_c(u,s;q|qZ)\mathcal{T}_c(u,s;q|q^{-1}Z)=\mathcal{T}_c(u,s;q|Z)^2+Z^{1/2}\mathcal{T}_c(uq,s;q|Z)\mathcal{T}_c(uq^{-1},s;q|Z).
\end{equation*}
Taking analytic continuation around $Z=0$ one can change the sign of $Z^{1/2}$ and get 
\begin{equation}\label{eq:Todac}
\mathcal{T}_c(u,s;q|qZ)\mathcal{T}_c(u,s;q|q^{-1}Z)=\mathcal{T}_c(u,s;q|Z)^2-Z^{1/2}\mathcal{T}_c(uq,s;q|Z)\mathcal{T}_c(uq^{-1},s;q|Z).
\end{equation}
\begin{prop}\label{contlimeq}
Substitute
\begin{equation}
Z=\hbar^4 z, \quad \mathcal{T}_c(u,s;q|Z)=\tau(\sg,s|z).
\end{equation}
If $\hbar\rightarrow 0$ then the leading order of the equation \eqref{eq:Todac} is Toda-like equation~\eqref{sgToda}.
\end{prop} 
Proof is by direct calculation, the first nontrivial coefficient appears in order $\hbar^2$.

\section{Formula for $\tau$ functions}
\label{sec:deform}
\subsection{$q$-deformed conformal blocks}
\label{ssec:qdefblocks}
In the representation theory approach irregular conformal block for $q$-deformed Virasoro algebra is defined as the square of Whittaker vector for this algebra (\cite{AY}).
This conformal block equals Nekrasov instanton partition function for 5d pure gauge $SU(2)$ theory \cite{NO} by proposed in \cite{AY} extension of the AGT conjecture. So the formula for the conformal block reads 
\begin{equation}
\calF(u_1,u_2;q_1,q_2|Z)=\sum_{\lambda_1,\lambda_2} Z^{|\lambda_1|+|\lambda_2|}\frac{1}{\prod_{i,j=1}^2 N_{\lambda_i,\lambda_j}(u_i/u_j;q_1,q_2)},
\label{confblock}
\end{equation}
where
\begin{equation}
N_{\lambda,\mu}(u,q_1,q_2)
=\prod_{s\in \lambda}
(1-u q_2^{-a_\mu(s)-1}q_1^{\ell_\lambda(s)}) \cdot 
\prod_{s \in \mu}
(1-u q_2^{a_\lambda(s)}q_1^{-\ell_\mu(s)-1}).\label{Nlm}
\end{equation}
The sum in \eqref{confblock} runs over all pairs of Young diagrams $\lambda_1, \lambda_2$. In formula \eqref{Nlm} $a_{\lmb}(s), l_{\lmb}(s)$ are lenghts of arms and legs 
of box $s$ in diagram $\lmb$.

The function $\mathcal{F}(Z)$ depends on $u_1,u_2$ through their ratio $u=u_1/u_2$, so we shall below use notation~$u$. Parameter $u$ is $q$-deformed analog of the highest weight and pair
$(q_1,q_2)$ is $q$-deformed analog of the central charge. In this paper we everywhere (except Appendix \ref{app:toralg}) use specification $q_1^{-1}=q_2=q$. In continuous limit this correspond to irregular conformal block for Virasoro algebras with $c=1$, as in formula for continuous $\tau$ function~\eqref{GIL}.

It follows directly from the definition that 
\begin{equation}
\mathcal{F}(u; q^{-1}, q|Z)=\mathcal{F}(u; q, q^{-1}|Z)=\mathcal{F}(u^{-1}; q^{-1}, q|Z). \label{blocksim}
\end{equation}

The function $\mathcal{F}$ is defined as a power series $\mathcal{F}(Z)=1+O(Z)$. This series converges, the proof of following proposition is similar to the one in \cite[Prop 1 (i)]{Its:2014}.

\begin{prop}\label{prop:convF}
	Let $|q|\neq 1$ and $u \neq q^n$, $n \in \mathbb{Z}$. Then series \eqref{confblock} converges uniformly and absolutely on every 	bounded subset of $\mathbb{C}$.
\end{prop}
\begin{proof}
	 There exist constants $L_1,L_2 \in \mathbb{R}_{>0}$ such that
	 \[\left|\frac{q^{n/2}-q^{-n/2}}{q^{1/2}-q^{-1/2}}\right|>|n| L_1^{1/2},\;\; \forall n \in \mathbb{Z}_{\neq 0}\qquad  \left|\frac{u^{1/2}q^{n/2}-u^{-1/2}q^{-n/2}}{q^{1/2}-q^{-1/2}}\right|>L_2^{1/2}, \;\; \forall n \in \mathbb{Z}.\]
	 Then we can bound
	   $\prod_{i,j=1}^2 N_{\lambda_i,\lambda_j}(u_i/u_j;q^{-1},q)$ as 
	 \begin{multline*}
	 	\Bigl| N_{\lambda_1,\lambda_1}(1;q^{-1},q) N_{\lambda_2,\lambda_2}(1;q^{-1},q) \Bigr|= \prod_{s \in \lambda_1}|q^{\frac12 h_{\lambda_1}(s)}-q^{-\frac12 h_{\lambda_1}(s)}|^2\\ \cdot \prod_{s \in\lambda_2}(\lambda_1 \leftrightarrow \lambda_2) > 
	 	 \frac{|\lambda_1|!^2|\lambda_2|!^2}{(\dim\lambda_1\dim\lambda_2)^2}\left|L_1(q^{1/2}-q^{-1/2})^2\right|^{|\lambda_1|+|\lambda_2|},
	 \end{multline*}
	 \begin{multline*}
	 \Bigl|N_{\lambda_1,\lambda_2}(u;q^{-1},q) N_{\lambda_2,\lambda_1}(u^{-1};q^{-1},q)\Bigr| =\prod_{s \in \lambda_1}|u^{\frac12}q^{\frac12(a_{\lambda_2}(s)+\ell_{\lambda_1}(s)+1)}-u^{-\frac{1}2}q^{-\frac12(a_{\lambda_2}(s)+\ell_{\lambda_1}(s)+1)}|^2\\ \cdot \prod_{s \in\lambda_2}(\lambda_1 \leftrightarrow \lambda_2, u\leftrightarrow u^{-1}) >   \left|L_2(q^{1/2}-q^{-1/2})^2\right|^{|\lambda_1|+|\lambda_2|},
	 \end{multline*}
	 where we used hook length formula for $\dim \lambda$.  Since $\sum_{|\lambda|=n} (\dim\lambda)^2=n!$ we have $\mathcal{F}(u;q|Z)<\exp\left|\dfrac{2|Z|}{L_1 L_2(q^{1/2}-q^{-1/2})^4}\right|$.
\end{proof}

To ensure convergence of infinite products like  $(u;q,q)_{\infty}=\prod_{i,j\geq 0}(1-uq^{i+j})$ we impose condition $|q|<1$. Using analytic continuation \eqref{qtrans} one can also work in the region $|q|>1$.

\begin{conj}\label{qHc}
The $q$-deformed conformal blocks satisfy  bilinear relations
\begin{multline}
 \sum_{2n\in\mathbb{Z}}  \frac{u^{2n}Z^{2n^2}}{\prod\limits_{\epsilon, \epsilon'=\pm1}(u^{\epsilon}q^{1+2\epsilon'n};q,q)_{\infty}}\kr{F}(uq^{-2n};q^{-1},q|q^{-1}Z)\kr{F}(uq^{2n};q^{-1},q|qZ)=\\
=(1-Z^{1/2})\sum_{2n\in\mathbb{Z}}  \frac{Z^{2n^2}}{\prod\limits_{\epsilon, \epsilon'=\pm1}(u^{\epsilon}q^{1+2\epsilon'n};q,q)_{\infty}}\kr{F}(uq^{-2n};q^{-1},q|Z)\kr{F}(uq^{2n};q^{-1},q|Z)
\label{bilconfrel}
\end{multline}
\end{conj}

We do not have proof of this conjecture, but we have two arguments in support. First this conjecture was checked by computer calculation up to $Z^4$ analytically and up to $Z^{12}$ numerically.

Second argument is the fact that continuous limit of this relation gives  known (\cite[eq. (4.29)]{KGP}) relation for irregular Virasoro conformal blocks $\mathcal{F}(\Delta|z)$
\begin{equation}
-2z^{1/2}\sum_{2n\in\mathbb{Z}} \frac{\mathcal{F}((\sg+n)^2|z) \mathcal{F}((\sg-n)^2|z)}{\prod_{k=1}^{2|n|-1}(k^2-4\sg^2)^{2(2|n|-k)}(4\sg^2)^{2|n|}} =
\sum_{2n\in\mathbb{Z}} \frac{D^2_{[\log z]}(\mathcal{F}(\sg+n)^2|z), \mathcal{F}(\sg-n)^2|z)}{\prod_{k=1}^{2|n|-1}(k^2-4\sg^2)^{2(2|n|-k)}(4\sg^2)^{2|n|}}
\label{bilincont}
\end{equation}
Here $D^2_{[\log z]}(f(z),g(z))=z^2(f''g-f'g'+fg'')+z(f'g+fg')$ denotes second Hirota
differential operator with respect to variable $\log z$.
We shall provide continuous limit and obtain \eqref{bilincont} in Subsection \ref{ssec:contlim}.

\begin{Remark}
Most results of this subsection can be stated for any $q_1, q_2$, see e.g. Appendix \ref{app:toralg} for the bilinear relations. But for $\tau$ functions we will use only conformal blocks with $q_1q_2=1$. 
\end{Remark}

\subsection{$q$-deformation of the formula for $\tau$ function}
\label{ssec:qdeftau}

\begin{defn}\label{Def:tau}
Function $\mathcal{T}(u,s;q|Z)$ given by the formula
\begin{equation}
\mathcal{T}(u,s;q|Z)=\sum_{n \in \mathbb{Z}}s^n C(uq^{2n};q|Z) \frac{\mathcal{F}(uq^{2n};q^{-1},q|Z)}{(uq^{2n+1};q,q)_\infty(u^{-1}q^{-2n+1};q,q)_\infty}\label{T}
\end{equation} 
is called $q$-deformed $\tau$ function of Painlev\'e III($D_8$) equation
if function $C(u;q|Z)$ satisfy equations
\begin{eqnarray}
\frac{C(uq;q|Z)C(uq^{-1};q|Z)}{C(u;q|Z)^2}&=&-Z^{1/2}\label{C01}\\
\frac{C(uq;q|qZ)C(uq^{-1};q|q^{-1}Z)}{C(u;q|Z)^2}&=&-uZ^{1/4}\label{C11}\\
\frac{C(u;q|qZ)C(u;q|q^{-1}Z)}{C(u;q|Z)^2}&=&Z^{-1/4}\label{C10}.
\end{eqnarray} 

If $C(u;q|Z)=C(u^{-1};q|Z)$ then
functions $C, \mathcal{T}$ are called $u$-inverse invariant.
%
\end{defn}

Evidently, $C(u;q|Z)$ could be multiplied on any function $\widetilde{C}(u;q|Z)$, which  satisfy homogeneous equations \eqref{C01}, \eqref{C11}, \eqref{C10}.

\begin{example}\label{Ex:C}
The following examples of $C(u;q|Z)$ satisfy \eqref{C01},\eqref{C11},\eqref{C10} 
\begin{eqnarray}
C_1(u;q|Z)&=&\Gamma((qZ)^{1/4};q^{1/4},q^{1/4})^3/\left(\Gamma(i(qZu)^{1/4};q^{1/4},q^{1/4})\,\Gamma(i(qZ)^{1/4}u^{-1/4};q^{1/4},q^{1/4})\right)
\\
C_c(u;q|Z)&=&(-1)^{2\left(\frac{\log u}{2\log q}\right)^2}\Gamma(-(qZ)^{1/4};q^{1/4},q^{1/4})\exp\left(\frac{\log ^2u\log Z}{4\log^2 q}\right).\label{exCc}
\end{eqnarray}
Here elliptic Gamma function is defined by $\Gamma(u;q,q)=(q^2 u^{-1};q,q)_{\infty}/(u;q,q)_{\infty}$. 
Necessary definitions and properties about $q$-deformed special functions are collected in Appendix \ref{app:q}.

Both functions $C_1,C_c$ are $u$-inverse invariant. Function $C_1$ is meromorphic as function on $Z^{1/2}$ (or on $u^{1/4}$) on complex plain. Function $C_c$ is useful to make continuous limit (Subsection \ref{ssec:contlim}). But $C_c$ is not meromorphic. 
\end{example}

\begin{Remark}
Left sides of the equations \eqref{C01}, \eqref{C11}, \eqref{C10} are multiplicative second order difference derivatives of
function on two variables $(u|Z)$ in directions $(1|0), (1|1), (0|1)$ correspondingly. 
\end{Remark}

\begin{Remark}
We could hide $s^n$ into the $C(uq^{2n};q|Z)$ carrying out some simple factor from the sum.
Indeed, function $\frac{\theta(us^{1/4}q;q)}{\theta(us^{-1/4}q;q)}$ satisfies homogeneous equations \eqref{C01},\eqref{C11}, \eqref{C10} so
\begin{equation*}
\mathcal{T}(u,s;q|Z)=\frac{\theta(us^{1/4}q;q)}{\theta(us^{-1/4}q;q)}\sum_{n\in\mathbb{Z}}C_s(uq^{2n};q|Z) \frac{\mathcal{F}(uq^{2n};q^{-1},q|Z)}{(uq^{2n+1};q,q)_\infty(uq^{-2n+1};q,q)_\infty},
\end{equation*}
where $C_s(u;q|Z)=C(u;q|Z) \frac{\theta(us^{-1/4}q;q)}{\theta(us^{1/4}q;q)}$.
\end{Remark}

It is clear from the definition that
\begin{equation}
\mathcal{T}(uq^2,s;q|Z)=s^{-1}\mathcal{T}(u,s;q|Z)\label{ushift}.
\end{equation}
If the function $C$ are $u$-inverse invariant then 
\begin{equation}
\mathcal{T}(u,s;q|Z)=\mathcal{T}(u^{-1},s^{-1};q|Z) \label{uinv}
\end{equation}
If $C$ is not $u$-inverse invariant then the function $\mathcal{T}(u^{-1},s;q|Z)$ also satisfies definition \eqref{T} but for another function $C$.

\begin{conj}\label{conj:Toda}
Function $\mathcal{T}(u,s;q|Z)$ satisfy bilinear equations
\begin{equation} 
Z^{1/4} \mathcal{T}(u,s;q|qZ) \mathcal{T}(u,s;q|q^{-1}Z) = \mathcal{T}(u,s;q|Z)^2+Z^{1/2}\mathcal{T}(uq,s;q|Z)\mathcal{T}(uq^{-1},s;q|Z). \label{qTodai}
\end{equation}
\end{conj}

\begin{thm}
Conjecture \ref{qHc} is equivalent to Conjecture \ref{conj:Toda}.  \label{mainbilin}
\end{thm}
\begin{proof}Let us substitute the expression \eqref{T} for $\tau$ function into \eqref{qTodai} and collect terms with the same powers of $s$.
The vanishing condition of the $s^m$ coefficient has the form
\begin{multline*}
Z^{1/4}\sum_{n\in\mathbb{Z}}C(uq^{2n+2m};q|qZ)C(uq^{-2n};q|q^{-1}Z) \frac{\kr{F}(uq^{2n+2m};q^{-1},q|qZ)\kr{F}(uq^{-2n};q^{-1},q|q^{-1}Z)}{\prod\limits_{\epsilon=\pm1}((uq^{2n+2m})^{\epsilon}q;q,q)_{\infty}((uq^{-2n})^{\epsilon}q;q,q)_{\infty}}=\\
=\sum_{n\in\mathbb{Z}}C(uq^{2n+2m};q|Z)C(uq^{-2n};q|Z) \frac{\kr{F}(uq^{2n+2m};q^{-1},q|Z)\kr{F}(uq^{-2n};q^{-1},q|Z)}{\prod\limits_{\epsilon=\pm1}((uq^{2n+2m})^{\epsilon}q;q,q)_{\infty}((uq^{-2n})^{\epsilon}q;q,q)_{\infty}}+\\
+Z^{1/2}\sum_{n\in\mathbb{Z}}C(uq^{2n+2m+1};q|Z)C(uq^{-2n-1};q|Z) \frac{\kr{F}(uq^{2n+2m+1};q^{-1},q|Z)\kr{F}(uq^{-2n-1};q^{-1},q|Z)}{\prod\limits_{\epsilon=\pm1}((uq^{2n+2m+1})^{\epsilon}q;q,q)_{\infty}((uq^{-2n-1})^{\epsilon}q;q,q)_{\infty}}.
\end{multline*}

Let us substitute into these relations $u\rightarrow uq^{-m}$, $n\rightarrow n-m/2$. We see that for any $m$ vanishing conditions
of the coefficients with powers $s^{2(m+\delta)}, \delta=0,1/2$  are equivalent to 
\begin{multline*}
Z^{1/4}\sum_{n\in\mathbb{Z}+\delta}C(uq^{2n};q|qZ)C(uq^{-2n};q|q^{-1}Z) \frac{\kr{F}(uq^{2n};q^{-1},q|qZ)\kr{F}(uq^{-2n};q^{-1},q|q^{-1}Z)}{\prod\limits_{\epsilon=\pm1}((uq^{2n})^{\epsilon}q;q,q)_{\infty}((uq^{-2n})^{\epsilon}q;q,q)_{\infty}}=\\
=\sum_{n\in\mathbb{Z}+\delta}C(uq^{2n};q|Z)C(uq^{-2n};q|Z) \frac{\kr{F}(uq^{2n};q^{-1},q|Z)\kr{F}(uq^{-2n};q^{-1},q|Z)}{\prod\limits_{\epsilon=\pm1}((uq^{2n})^{\epsilon}q;q,q)_{\infty}((uq^{-2n})^{\epsilon}q;q,q)_{\infty}}+\\
+Z^{1/2}\sum_{n\in\mathbb{Z}+\delta}C(uq^{2n+1};q|Z)C(uq^{-2n-1};q|Z) \frac{\kr{F}(uq^{2n+1};q^{-1},q|Z)\kr{F}(uq^{-2n-1};q^{-1},q|Z)}{\prod\limits_{\epsilon=\pm1}((uq^{2n+1})^{\epsilon}q;q,q)_{\infty}((uq^{-2n-1})^{\epsilon}q;q,q)_{\infty}}.
\end{multline*}
We want to divide these conditions on $C(u;q|Z)^2$ and then simplify arising expressions like
$\frac{C(\cdot;q|\cdot)C(\cdot;q|\cdot)}{C(\cdot;q|\cdot)^2}$.
Introduce auxiliary functions $\beta_k$, $\gamma_k$, $k\in \mathbb{Z}$ by formulas
\begin{equation}
\beta_k=\frac{C(uq^{k};q|Z)C(uq^{-k};q|Z)}{C(u;q|Z)^2}, \quad \gamma_k=\frac{C(uq^{k};q|qZ)C(uq^{-k};q|q^{-1}Z)}{C(u;q|qZ)C(u;q|q^{-1}Z)}.
\end{equation}
Elementary calculations give us
\begin{equation}
\beta_k=Z\beta^2_{k-1}/\beta_{k-2}, \quad \gamma_k=Z\gamma^2_{k-1}/\gamma_{k-2}. \label{bgd2}
\end{equation}
Evidently $\beta_0=\gamma_0=1$. From \eqref{C01}, \eqref{C11}, \eqref{C10} we have $\beta_1=-Z^{1/2}, \gamma_1=-uZ^{1/2}$.
So we have unique solution of \eqref{bgd2}
\begin{equation}
\beta_k=(-1)^kZ^{k^2/2},\quad \gamma_k=(-1)^ku^kZ^{k^2/2}.
\end{equation}
Using these results and \eqref{C10} we obtain
\begin{multline*}
\sum_{n\in\mathbb{Z}+\delta}u^{2n}Z^{2n^2} \frac{\kr{F}(uq^{2n};q^{-1},q|qZ)\kr{F}(uq^{-2n};q^{-1},q|q^{-1}Z)}{\prod\limits_{\epsilon=\pm1}((uq^{2n})^{\epsilon}q;q,q)_{\infty}((uq^{-2n})^{\epsilon}q;q,q)_{\infty}}=\\
=\sum_{n\in\mathbb{Z}+\delta}Z^{2n^2} \frac{\kr{F}(uq^{2n};q^{-1},q|Z)\kr{F}(uq^{-2n};q^{-1},q|Z)}{\prod\limits_{\epsilon=\pm1}((uq^{2n})^{\epsilon}q;q,q)_{\infty}((uq^{-2n})^{\epsilon}q;q,q)_{\infty}}-\\
-Z^{1/2}\sum_{n\in\mathbb{Z}+\delta}Z^{(2n+1)^2/2} \frac{\kr{F}(uq^{2n+1};q^{-1},q|Z)\kr{F}(uq^{-2n-1};q^{-1},q|Z)}{\prod\limits_{\epsilon=\pm1}((uq^{2n+1})^{\epsilon}q;q,q)_{\infty}((uq^{-2n-1})^{\epsilon}q;q,q)_{\infty}},
\end{multline*}
Substituting in last sum $n\rightarrow n-1/2$ we obtain
\begin{multline*}
\sum_{n\in\mathbb{Z}+\delta}u^{2n}Z^{2n^2} \frac{\kr{F}(uq^{2n};q^{-1},q|qZ)\kr{F}(uq^{-2n};q^{-1},q|q^{-1}Z)}{\prod\limits_{\epsilon=\pm1}((uq^{2n})^{\epsilon}q;q,q)_{\infty}((uq^{-2n})^{\epsilon}q;q,q)_{\infty}}=\\
=\sum_{n\in\mathbb{Z}+\delta}Z^{2n^2} \frac{\kr{F}(uq^{2n};q^{-1},q|Z)\kr{F}(uq^{-2n};q^{-1},q|Z)}{\prod\limits_{\epsilon=\pm1}((uq^{2n})^{\epsilon}q;q,q)_{\infty}((uq^{-2n})^{\epsilon}q;q,q)_{\infty}}-\\
-Z^{1/2}\sum_{n\in\mathbb{Z}+\delta+1/2}Z^{2n^2} \frac{\kr{F}(uq^{2n};q^{-1},q|Z)\kr{F}(uq^{-2n};q^{-1},q|Z)}{\prod\limits_{\epsilon=\pm1}((uq^{2n})^{\epsilon}q;q,q)_{\infty}((uq^{-2n})^{\epsilon}q;q,q)_{\infty}},
\end{multline*}

And this is a result of splitting \eqref{bilconfrel} into part with integer powers of $Z$ and with half-integer powers of $Z$.
\end{proof}



\subsection{Continuous limit of $q$-deformed $\tau$ function}
\label{ssec:contlim}
\begin{prop}
	Formula for $q$-deformed $\tau$ function $\mathcal{T}(u,s;q|Z)$ could be rewritten in next way
	\begin{equation}
	\mathcal{T}(u,s;q|Z)={C}(u;q|Z)\sum_{n\in\mathbb{Z}} \td{s}(u,s;q|Z)^n Z^{n^2+n/2}
	\frac{\mathcal{F}(uq^{2n};q^{-1},q|Z)}{\prod\limits_{\epsilon=\pm1}((uq^{2n})^{\epsilon}q;q,q)_{\infty}}, \label{Tr}
	\end{equation}
	where
	\begin{equation}
	\td{s}(u,s;q|Z)=-\left(\frac{{C}(u;q|Z)}{{C}(uq^{-1};q|Z)}\right)^2s 
	\end{equation}
\end{prop}
\begin{proof}
	
	Using \eqref{C01} we transform expression $C(uq^{k};q|Z)/C(u;q|Z)$ as
	\begin{align*}
	\frac{C(uq^{k};q|Z)}{C(uq^{k-1};q|Z)}&=\frac{C(uq^{k-1};q|Z)}{C(uq^{k-2};q|Z)}\frac{C(uq^k;q|Z)C(uq^{k-2};q|Z)}{C(uq^{k-1};q|Z)^2}=
	-\frac{C(uq^{k-1};q|Z)}{C(uq^{k-2};q|Z)}Z^{1/2}\Rightarrow\\
	\frac{C(uq^{k};q|Z)}{C(uq^{k-1};q|Z)}&=(-1)^{k}Z^{\frac{k}2}\frac{C(u;q|Z)}{C(uq^{-1};q|Z)}
	\Rightarrow\\
	\frac{C(uq^{k};q|Z)}{C(uq^{-1};q|Z)}&=(-1)^{\frac{k(k+1)}2}Z^{\frac{k(k+1)}{4}}\left(\frac{C(u;q|Z)}{C(uq^{-1};q|Z)}\right)^{k+1}.
	\end{align*}
Substituting the last expression to \eqref{T} we finish the proof.
\end{proof}

\begin{Remark}
	The $\tau$-function was defined as a series and it is convenient to prove convergence of this series using expression \eqref{Tr}. 
	The proof is analogous to one in \cite[Prop. 1 (ii)]{Its:2014}.
	
	First, note that bounds $L_1, L_2$ defined in the proof of Proposition \ref{prop:convF} are the same for every conformal block in the sum \eqref{Tr}. Therefore we can estimate conformal blocks by the same exponent. 
	
	Second we to rewrite Pochhammer symbols in terms
	of $q$-Barnes $\G$-function using \eqref{GN}. Then for $n>0$ we have 
	\begin{equation*}
	\G(1+2(\sg+n);q)\G(1-2(\sg+n);q)=\frac{\G(1-2\sg;q)}{\G(1+2\sg;q)} \frac{(-1)^n\theta(uq;q)^{2n}} {u^{n(2n-1)} q^{\frac{n(4n^2-1)}3} (1-q)^{2n} (q;q)^{4n}_{\infty}}\G^2(1+2(\sg+n);q).
	\end{equation*}
	where $\sg=\frac{\log u}{2\log q}$. It follows form  \cite[Prop. 3.1]{Nish} that $\log \G(x;q)\sim -\log(1-q)x^2/2$,  $\operatorname{Re}(x)\rightarrow+\infty$ (this is the first term in the sum in loc. cit. and one can show the all other terms are majorized by it). Using $|q|<1$ and this asymptotic behavior   we see that this coefficients dominate $Z^{n^2+n/2}\tilde{s}^n$.
	
	Calculation for $n<0$ is the same with the replacement $\sigma \leftrightarrow -\sigma$.
\end{Remark}

\begin{Remark}
	Moreover, we could rewrite $\mathcal{T}(u,s;q|Z)$ as sum of conformal blocks with some $\hat{s}$ and $q$-rational coefficients 
	(in analogous to continuous $\tau$ function (\cite{BSRamon})).
	Introduce functions $P_n(u;q)$ by 
	\begin{equation*}
	\frac{\prod\limits_{\epsilon=\pm1}(u^{\epsilon}q;q,q)_{\infty}}{\prod\limits_{\epsilon=\pm1}((uq^{2n})^{\epsilon}q;q,q)_{\infty}}
	=P_n(u;q)\left(\frac{(u;q)_{\infty}}{(u^{-1};q)_{\infty}}\right)^{2n},\quad n \in \mathbb{Z}
	\end{equation*}
	Then we can write  
	\begin{equation}\label{eq:Trational}
	\mathcal{T}(u,s;q|Z)=\frac{C(u;q|Z)}{\prod\limits_{\epsilon=\pm1}(u^{\epsilon}q;q,q)_{\infty}}\sum_{n\in\mathbb{Z}} Z^{n^2+n/2} \hat{s}^n P_n(u;q)
	\mathcal{F}(uq^{2n};q^{-1},q|Z),
	\end{equation}
	\begin{align*}
	\text{where}\quad 	 &P_n(u;q)=
	\frac{(-1)^n}{(1-u)^{2n}\prod_{i=1}^{2n-1}(u^{1/2}q^{i/2}-u^{-1/2}q^{-i/2})^{2(2n-i)}},\;\; n\geq 0,
	\\
	&P_n(u;q)= P_{-n}(u^{-1};q), \;\; n<0,\qquad  \hat{s}=-\left(\frac{C(u;q|Z)}{C(uq^{-1};q|Z)}\frac{(u;q)_{\infty}}{(u^{-1};q)_{\infty}}\right)^2s.
	\end{align*}
	The formula \eqref{eq:Trational} can also be used for the proof of convergence \eqref{T}. The proof goes similarly to the proof of Proposition \ref{prop:convF}, we bound all terms in the denominator of $P_n(u;q)$ using $L_1,L_2$.
\end{Remark}

Now we check that in $q \rightarrow 1^-$ limit (i.e. $q\in [1-\epsilon,1]$) the formula for   $\mathcal{T}(u,s;q|Z)$ gives the formula \eqref{GIL} for $\tau$ function of continuous Painlev\'e equation. As in Subsection \ref{ssec:contlimeq} it is convenient to use function $\mathcal{T}_c(u,s;q|Z)$ defined by \eqref{redqToda}. Moreover for this limit we take $q$-deformed $\tau$ functions with $C(u;q|Z)=C_c(u;q|Z)$ given by formula \eqref{exCc}.

\begin{thm}
\label{thmlim}
Let 
\begin{equation}
q=e^{\hbar},\quad Z=\hbar^4 z, \quad \sg=\frac{\log u}{2\hbar}. 
\end{equation}
Then the $\tau$-function $\mathcal{T}_c(u,(-1)^{-4\sigma}s;q|Z)$ with $C(u;q|Z)=C_c(u;q|Z)$ goes to $\tau(\sg,s|z)$ in the limit $\hbar\rightarrow -0$.
\end{thm}
\begin{proof}
Rewrite $\mathcal{T}_c(u,s;q|Z)$ in form \eqref{Tr} using \eqref{exCc}
\begin{equation}
\mathcal{T}_c(u,s;q|Z)=(-1)^{2\sg^2}\sum_{n\in\mathbb{Z}}  Z^{(\sg+n)^2} ((-1)^{4\sg}s)^n
\frac{\mathcal{F}(uq^{2n};q^{-1},q|Z)}{\prod\limits_{\epsilon=\pm1}((uq^{2n})^{\epsilon}q;q,q)_{\infty}}\label{toconttau}
\end{equation}

First, we prove convergence of conformal blocks $\mathcal{F}(u;q^{-1},q|Z) \rightarrow \mathcal{F}(\sigma^2;|z)$. For each summand in \eqref{confblock} we have 
\begin{equation*}
N_{\lmb,\mu}(q^{-1},q,u)\sim(-\hbar)^{|\lmb|+|\mu|}\prod_{s\in\lmb}(2\sg-(a_{\mu}(s)+l_{\lmb}(s)+1))\times  \prod_{s\in\mu}(2\sg+(a_{\mu}(s)+l_{\lmb}(s)+1))
\end{equation*}
This occasional power of $\hbar$ cancels due to definition $Z=\hbar^4z$. 

Next, we prove that convergence of series \eqref{confblock} is uniform on $q \in [1-\epsilon,1]$. Note that in this region $\left|\frac{q^{x/2}-q^{-x/2}}{q^{1/2}-q^{-1/2}}\right|\geq x$, for $x \geq 1$ so $L_1$ (from the proof of Prop \ref{prop:convF}) is uniformly bounded by $L_1\geq 1$. And for $L_2$ using evident inequalities \[|q^{y}-q^{-y}|\geq |q^{\operatorname{Re} y}-q^{-\operatorname{Re}y}|,\quad |q^{a_1}-q^{-a_1}|\geq |q^{a_2}-q^{-a_2}|, \quad  \left|\frac{q_1^{a}-q_1^{-a}}{q_1^{1/2}-q_1^{-1/2}}\right|\leq \left|\frac{q_2^{a}-q_2^{-a}}{q_2^{1/2}-q_2^{-1/2}}\right|, \]
for $a_1>a_2>0$ and $0<q_1<q_2<1$, $0<a<1/2$. Therefore we have bound 
\[L_2^{1/2}\geq \frac{|(1-\epsilon)^{\operatorname{Re} \sigma+n_0/2}-(1-\epsilon)^{-\operatorname{Re}\sigma -n_0/2}|}{(1-\epsilon)^{1/2}-(1-\epsilon)^{-1/2}},\]
where $-1/4<\operatorname{Re} \sigma+n_0/2<1/4$. Since bounds for $L_1, L_2$ are uniform we prove that convergence of \eqref{confblock} are also uniform.

Next we study coefficients in \eqref{toconttau}. Using 
formula \eqref{GN} we can rewrite 
\begin{equation}
Z^{(\sg+n)^2}\frac{1}{\prod\limits_{\epsilon=\pm1}((uq^{2n})^{\epsilon}q;q,q)_{\infty}}
 =\frac{1}{\mathsf{G}(1+2(\sg+n);q)\mathsf{G}(1-2(\sg+n);q)} \left(\frac{Z}{(1-q)^4}\right)^{(\sg+n)^2}.
\end{equation}
		Then using Theorem \ref{Glim} and $\frac{Z}{(1-q)^4}\sim z$ we that coefficients goes to the coefficients in $\tau$ function for continuous Painlev\'e equation \eqref{GIL}. 

It remains to show that series \eqref{toconttau} converges uniformly on $q \in [1-\epsilon,1]$. This can be done similar to the series \eqref{confblock} above. 
\end{proof}

It follows from this theorem, Theorem \ref{mainbilin} and Proposition \ref{contlimeq} that $q \rightarrow 1$ limit of bilinear relations \eqref{bilconfrel} is \eqref{bilincont}. This fact was stated in the end of Subsection \ref{ssec:qdefblocks}.

Note that Bonelli, Grassi and Tanzini in their paper \cite{BGT} also
constructed function, which in the limit $q \rightarrow 1$ goes to the
$\tau(\sigma,s|z)$. It is interesting to note that they work in
different region of $q$, namely $|q|=1$ in their paper. The relation between our results and \cite{BGT} will be studied in \cite{BGT2}.

%

\subsection{$q$-deformation of Painlev\'e III($D_8$) algebraic solution}
\label{ssec:algsol}

For special values of parameters $\sg=1/4, s=\pm 1$ the sum in formula \eqref{GIL} can be calculated (see \cite[Sec. 3.3]{BSRamon})
\begin{equation}
\tau(1/4,\mp 1|z)=\frac{1}{\mathsf{G}(1/2)\mathsf{G}(3/2)}z^{1/16}e^{\pm4z^{1/2}} \label{stattau}
\end{equation}
These $\tau$ functions correspond to the algebraic solutions of Painlev\'e III($D_8$) equation \eqref{Pp} $w(z)=\mp z^{1/2}$. These solutions are invariant under B\"acklund transformation so
$\tau(\sg,s|z)\propto \tau(\sg+1/2,s|z)$ and substituting this into \eqref{sgToda} we obtain equation
\begin{equation}
1/2 D^2_{[\log z]}(\tau(z),\tau(z))=\pm z^{1/2}\tau(z)^2. \label{statToda}
\end{equation}
This equation gives us \eqref{stattau} up to multiplying on constant and any power of $z$. 

We want to obtain analogous relation on $q$-deformed conformal blocks.
Recall that element $\pi_2^2\in W$ is an analogue of B\"acklund transformation and acts as  $\mathcal{T}(u,s;q|Z) \mapsto \mathcal{T}(uq,s;q|Z)$. Therefore if $\tau$ function $\mathcal{T}(u,s;q|Z)$ is $u$-inverse invariant then for $u=q^{1/2}$, $s=\pm 1$ we have 
\begin{equation}
\pi_2^2(\mathcal{T}(q^{1/2},\pm 1;q|Z))=\mathcal{T}(q^{3/2},\pm 1;q|Z) =\pm \mathcal{T}(q^{-1/2},\pm 1;q|Z)=\pm \mathcal{T}(q^{1/2},\pm 1;q|Z), \label{eq:q12shifts}
\end{equation}
where we used \eqref{ushift}.

The corresponding function $G(z)=\mp Z^{1/2}$ due to \eqref{FGdefinition} and \eqref{eq:T1T3}. These $G(Z)$ are  algebraic (and B\"acklund invariant) solutions of \eqref{qPp}.
\begin{conj}\label{Pochh}
For any $u$-inverse invariant $\tau$ function $\mathcal{T}(u,s;q|Z)$ we have
\begin{equation}
\mathcal{T}(q^{1/2},\pm 1;q|Z)=\frac{C(q^{1/2};q|Z)}{(q^{3/2};q,q)_{\infty}(q^{1/2};q,q)_{\infty}}(\mp Z^{1/2}q^{1/2};q^{1/2},q^{1/2})_{\infty}. \label{tau14}
\end{equation}
Equivalently, we have relation on $q$-deformed conformal blocks 
\begin{equation}\label{eq:algebconfblock}
(\mp Z^{1/2}q^{1/2};q^{1/2},q^{1/2})_{\infty}=\sum_{n\in\mathbb{Z}}(\mp 1)^n Z^{n^2+n/2}P_n(q)
\mathcal{F}(q^{2n+1/2},q,q|Z),
\end{equation}
where
\begin{equation}
P_n(q)=\frac{\prod\limits_{\epsilon=\pm1}(q^{\frac12\epsilon}q;q,q)_{\infty}}{\prod\limits_{\epsilon=\pm1}((q^{2n+1/2})^{\epsilon}q;q,q)_{\infty}}=\prod_{j=0}^{k-1}\frac{1}{\left((1-q^{j+1/2})(1-q^{-j-1/2})\right)^{k-j}},
\end{equation}
where $k=2n$, for $n>0$ and $k=-2n-1$, for $n<0$.
\end{conj}
One can compare \eqref{eq:algebconfblock} with \eqref{eq:Trational} and see that $P_n(q)=(-1)^nP_n(q^{1/2},q)$.
\begin{thm}
If Conjecture \ref{qHc} holds then Conjecture \ref{Pochh} also holds.\label{Pochht}
\end{thm}

\begin{proof}
Due to Theorem \ref{mainbilin} we can use Conjecture \ref{conj:Toda} instead of Conjecture \ref{qHc}. Using \eqref{eq:q12shifts} we can rewrite equation \eqref{qTodai} as 
\begin{equation}\label{eq:Todaalg}
Z^{1/4} \mathcal{T}(q^{1/2},\pm 1;q|qZ) \mathcal{T}(q^{1/2},\pm 1;q|q^{-1}Z) =(1\pm Z^{1/2}) \mathcal{T}(q^{1/2},\pm 1;q|Z)^2.
\end{equation}
On the other hand, using \eqref{Tr} we have 
\[
\mathcal{T}(q^{1/2},\pm 1;q|qZ)={C}(q^{1/2};q|Z)\sum_{n\in\mathbb{Z}} (\mp 1)^n Z^{n^2+n/2}
\frac{\mathcal{F}(q^{2n+1/2};q^{-1},q|Z)}{\prod\limits_{\epsilon=\pm1}((uq^{2n})^{\epsilon}q;q,q)_{\infty}}, 
\]
where we used $u$-inverse invariance of $C(u;q|Z)$. Denote 
\[
f(z)=\mathcal{T}(q^{1/2},\pm 1;q|qZ) \frac{(q^{3/2};q^{1},q^{1})_{\infty}(q^{1/2};q,q)_{\infty}}{C(q^{1/2};q|Z)}
\]
Then $f(z)$ is equal to the right side of \eqref{eq:algebconfblock}. And due to \eqref{C10} the equation \eqref{eq:Todaalg} for this function takes the form 
\begin{equation}
f(qZ) f(q^{-1}Z)=(1\pm Z^{1/2}) f(Z)f(Z).
\end{equation}
The only solution $f(Z)=\sum_{n \geq 0} f_n Z^{n/2}$ with $f_0=1$ of the last equation  is $f(Z)= (\mp Z^{1/2}q^{1/2};q^{1/2},q^{1/2})_{\infty}$.
\end{proof}

We have checked equality \eqref{eq:algebconfblock} up to $Z^4$. 

\begin{Remark}
Let us check the continuous limit of \eqref{tau14}. We want to use Theorem \ref{thmlim} so we set $C(u;q|Z)=C_c(u;q|Z)$. Let 
Then we divide	 \eqref{tau14} by $\Gamma(-(qZ)^{1/4};q^{1/4},q^{1/4})/(q;q,q)^2_{\infty}
$ and get $\mathcal{T}_C$ in the left side (see \eqref{redqToda}).  Using Theorem \ref{thmlim} we have that if $Z=\hbar^4z,\, q=e^{\hbar}$, $q\rightarrow 1^-$ then $\mathcal{T}_c(q^{1/2},\pm1;q|Z)\rightarrow\tau(1/4,\mp 1|z).$

For the continuous limit of $q$-Pochhammer symbol in right side of \eqref{tau14} we have (using \eqref{Pexpr})
\begin{equation}
(\mp Z^{1/2}q^{1/2};q^{1/2},q^{1/2})_{\infty}=\exp\left(-\sum_{m=1}^{\infty}\frac{(\mp Z^{1/2}q^{1/2})^m }{m(1-q^{m/2})^2}\right)\sim e^{\pm 4z^{1/2}},
\end{equation}
where we observed that only the term of sum with $m=1$ survives.
Using this result and Theorem~\ref{Glim} we obtain that limit \eqref{tau14} is known relation \eqref{stattau} (after analyti continuation arownd $Z=0$ as in Subsection \ref{ssec:contlimeq}).
\end{Remark}

\section{Further questions} \label{sec:concl}

\begin{itemize}
 \item The main statements of the paper are based on the conjecture 
 \ref{qHc}, so the first question is to prove it. 
 
 \item The second question is about a generalization of our results to the other discrete Painlev\'e equations from Sakai's tables Fig \ref{Fig:Sakaisurf}, \ref{Fig:Sakaisym}. It is natural to conjecture that $q$-difference Painlev\'e equations $A_{7-N}^{(1)}/E_{N+1}^{(1)}$ surface/symmetry type $N \leqslant 7$  are related to the Nekrasov partition functions for 5d $SU(2)$ gauge theory with $N$ fundamental multiplets. 
 
 It was argued by Seiberg \cite{Seib} that these gauge theories has $E_{N+1}$ global symmetry. It would be interesting to find a physical interpretation of the affine Weyl groups $E_{N+1}^{(1)}$, the symmetry group of the corresponding discrete Painlev\'e equations.
  
 For $N \leq 4$ one can take $q\rightarrow 1$ limit and get the relation between differential Painlev\'e equation with surface type $D_4^{(1)}, D_5^{(1)}, D_6^{(1)}, D_7^{(1)}, D_8^{(1)}$ and Nekrasov partition functions for 4d $SU(2)$ gauge theory with $4, 3, 2, 1, 0$ fundamental multiplets correspondingly. This relation  was stated in \cite{GIL1302} and proven in \cite{ILT} and \cite{KGP}.
 
 \item Elements $T, \pi_2^2 \in W$ act on the $\tau$ function $\mathcal{T}(u,s;q|Z)$ defined in \eqref{T} in a clear manner
 \[ T\colon \mathcal{T}(u,s;q|Z) \mapsto \mathcal{T}(u,s;q|qZ),\quad  \text{and}\quad \pi_2^2 \colon \mathcal{T}(u,s;q|Z) \mapsto \mathcal{T}(u q,s;q|Z) \] 
 It is natural to ask for the action of whole group $W$. From the Table \ref{lettertable} we see that remaining transformations are $Z \mapsto Z^{-1}$ and $q \mapsto q^{-1}$. The second transformation is transparent due to $q$-deformed conformal block relation \eqref{blocksim} and $q$ Pochhammer symbol relation \eqref{qtrans}. But it is unclear what is (if exists) the meaning of $Z \mapsto Z^{-1}$ symmetry for the formula \eqref{T}. In particular it is unclear what is the form of the corresponding relation between $q$-deformed conformal blocks (Nekrasov partition functions for pure theory) $\mathcal{F}(Z)$ and $\mathcal{F}(Z^{-1})$.
 
 One more remark is in order. Nekrasov partition function for pure 5d $SU(2)$ gauge theory is equal (up to simple factor) to the topological string partition function for local $\mathbb{P}^1\times \mathbb{P}^1$ geometry~\cite{Igbal:2003},~\cite{Eguchi:2003}. For such partition functions there exists fiber-base duality \cite{Mitev} which interchange two factors $\mathbb{P}^1$. In terms of the functions $\mathcal{F}(u;q^{-1},q|Z^{-1})$ this duality has the form 
 \begin{equation}\label{Zinv}
 \frac{\mathcal{F}(u; q^{-1}, q|Z)}{(uq;q,q)_{\infty}^2} =
 \frac{\mathcal{F}(uZ; q^{-1}, q|Z^{-1})}{(uZq;q,q)_{\infty}^2}. 
 \end{equation}
 But this is an equality of formal power series on variables $u, uZ$. For $|q|\neq 1$ function $\mathcal{F}(u; q^{-1}, q|Z)$ cannot be expanded as such convergent series since has poles for $u=q^n$ and these poles accumulate near $u=0$. Moreover, computer calculations shows that \eqref{Zinv} does not hold for the function $\mathcal{F}(u; q^{-1}, q|Z)$ defined by convergent series \eqref{confblock}. It is an interesting question whether exist relation of the type \eqref{Zinv} for $|q|\neq 1$ and can it be used to the $Z \mapsto Z^{-1}$ transformation of $\tau$-function $\mathcal{T}(u,s;q|Z)$.

  \item Continuous Painlev\'e equations can be described as 
   non-autonomous Hamiltonian system. It is natural to ask for the difference analogue of this fact. 
   
  Painlev\'e III($D_8$) has two types of bilinear forms, namely Toda-like and Okamoto-like (see \cite{BSRamon}). In Subsection \ref{ssec:contlimeq}  we give $q$-deformation of the Toda-like bilinear form, it is natural to ask for the $q$-deformation of Okamoto-like equations.
  
  \item It is interesting to note that there exists another $q$-difference equation wichi has the Painleve~III$(D_8)$ equation in the $q\rightarrow 1$ limit. This equation has the form 
  \begin{equation}\label{eq:qA7}
  \overline{W}W^2\underline{W}=Z(1-W).
  \end{equation}
  and the limit was shown in \cite{Grammaticos:2002} (see also \cite{Ramani:2015}) In terms of Sakai classification this equation corresponds to $A_7^{(1)}$ surface (see \cite[eq. (2.44)]{SakaiLax}), so should not be equivalent to \eqref{qPp}. It is an~interesting question whether exists a way to express solutions of \eqref{eq:qA7} in terms of $q$-deformed conformal blocks.
\end{itemize}

\section{Acknowledgments}
We thank  V. Adler, G. Bonelli, B. Feigin, A. Grassi, P. Gavrylenko, N. Iorgov, O. Lisovyy, A. Sciarappa, V. Spiridonov, K.~Sun, A. Tanzini, Y. Zenkevich for interest in our work and discussions. We are grateful to A. Dzhamay, K. Kajiwara, H. Sakai and T. Takenawa for many explanations about discrete Painleve equations. It was H. Sakai who explained that we should use $A_7^{(1)\prime}$ surface for $q$-deformation of PIII$(D_8)$ equation. M. B. is grateful to INFN Trieste and SISSA for the hospitality during his visit of Italy and to T. Takenawa for the hospitality during his visit of Japan.

This work has been funded by the Russian Academic Excellence Project
’5-100’.  A.S. was also supported in part by joint NASU-CNRS project F14-2016, M.B. was also supported in part  by Young Russian Mathematics award and RFBR grant mol\_a\_ved 15-32-20974.
Study of $q$-deformed bilinear relations was performed under a grant of Russian Science Foundation (project No. 14-12-01383).

\appendix

\section{$q$-special functions}
\label{app:q}
In this Appendix we collect some facts about $q$-series that we used in the paper. For the references see \cite[Sec. 10]{AAR}, \cite{Nish}, \cite{Spir}.

Infinite multiple $q$-deformed Pochhammer symbol is defined by
\begin{equation}
(Z;t_1,\ldots t_N)_{\infty}=\prod_{i_1,\ldots i_N=0}^{\infty}\left(1-Z\prod_{k=1}^Nt_k^{i_k}\right).
\end{equation}
Product exists if all $|t_k|<1$. Function is symmetric with respect to $t_k$. In this region function is analytic function of all arguments.
Infinite $q$-Pochhammer symbols satisfy
\begin{equation}
(Z;t_1,\ldots t_N)_{\infty}/(Zt_1;t_1,\ldots t_N)_{\infty}=(Z|t_2,\ldots t_N)_{\infty}, \quad (Z;q)_{\infty}/(Zq;q)_{\infty}=1-Z \label{shift} 
\end{equation}

The function $(Z;t_1,\ldots t_N)_{\infty}$ can be rewritten as
\begin{multline}
(Z;t_1,\ldots t_N)_{\infty}=\exp\left(\sum_{i_1,\ldots i_N=0}^{\infty}\log\left(1-Z\prod_{k=1}^Nt_k^{i_k}\right)\right)=\\=
\exp\left(-\sum_{i_1,\ldots i_N=0}^{\infty}\sum_{m=1}^{\infty}\frac{Z^m}m\prod_{k=1}^Nt_k^{mi_k}\right)
=\exp\left(-\sum_{m=1}^{\infty}\frac{Z^m}m\prod_{k=1}^N\frac1{1-t_k^m}\right),
\label{Pexpr}
\end{multline}
where sum converge when $|Z|<1, |t_k|\neq 1$.
Using this expression the function $(Z;t_1,\ldots t_N)_{\infty}$ can be defined to the region with some $|t_k|$ greater then $1$.
Using this definition we see that
\begin{equation}
(Z;t_1^{-1},t_2,\ldots t_N)_{\infty}=(Zt_1;t_1,\ldots t_N)^{-1}_{\infty} \label{qtrans}
\end{equation}
In the paper we use only $N=1,2$ $q$-Pochhammer symbols.

Introduce trigonometric (or $q$-) Gamma function and trigonometric (or $q$-) Barnes $\mathsf{G}$ function according to \cite{Nish}
\begin{eqnarray}
\Gamma(x;q)&=&(1-q)^{1-x}\frac{(q;q)_{\infty}}{(q^x;q)_{\infty}},\\
\mathsf{G}(x;q)&=&(1-q)^{-\frac{(x-1)(x-2)}2}\prod\limits_{k=0}^{\infty}\frac{(1-q^{x+k})^{k+1}}{(1-q^{k+1})^{k+2-x}}=
(1-q)^{-\frac{(x-1)(x-2)}2}\frac{(q;q)^{x-1}_{\infty}(q^x;q,q)_{\infty}}{(q;q,q)_{\infty}}.\label{GN}
\end{eqnarray}
From \eqref{shift} we have
\begin{align}
\Gamma(u+1;q)&=[u]_q\Gamma(u;q),\;\textrm{ where } [u]_q=\frac{1-q^u}{1-q}\\
\mathsf{G}(u+1;q)&=\Gamma(u;q)\mathsf{G}(u;q).
\end{align}
Function $[u]_q$ is called $q$-number.

We use part of the Theorem 4.4 from \cite{Nish}
\begin{thm}
As $q\rightarrow 1^-$ $\Gamma(u;q)$ and $\mathsf{G}(u;q)$ converge to $\Gamma(u)$ and $\mathsf{G}(u)$.
Convergence is uniform on any compact set in the domain $\mathbb{C}\backslash \mathbb{Z}_{\leq 0}$.\label{Glim}
\end{thm}

Introduce elliptic Gamma function (see e.g. \cite{Spir})
\begin{equation}
\Gamma(Z;t,q)=\frac{(tqZ^{-1};t,q)_{\infty}}{(Z;t,q)_{\infty}}.\label{GPoch}
\end{equation}
Elliptic gamma function should not be confused with trigonometric gamma and Barnes $\mathsf{G}$-function with one deformation argument and with standard Gamma function
and Barnes $\mathsf{G}$-function without deformation arguments.
Elliptic Gamma function satisfy relations
\begin{equation}
\Gamma(qZ;t,q)=\theta(Z;t)\Gamma(Z;t,q), \quad \Gamma(tZ;t,q)=\theta(Z;q)\Gamma(Z;t,q),\label{Gshift}
\end{equation}
where $\theta$ function defined as
\begin{equation}
\theta(Z;q)=\frac{1}{(q;q)_{\infty}}\sum_{k\in\mathbb{Z}}(-1)^k q^{\frac{k(k-1)}{2}}Z^k=(Z;q)_{\infty}(qZ^{-1};q)_{\infty}.
\end{equation}
Here last equality is Jacobi triple product.
Shift relations \eqref{Gshift} on elliptic Gamma functions could be easily obtained from \eqref{shift}.
It follows from definition that $\theta$ function satisfies
\begin{equation}
\theta(qZ;q)=-Z^{-1}\theta(Z;q)=\theta(Z^{-1};q). \label{tshift}
\end{equation}
From \eqref{Gshift} and first equality \eqref{tshift} we obtain useful relation
\begin{equation}
\frac{\Gamma(uq,q;q)\Gamma(uq^{-1};q,q)}{\Gamma(u;q,q)^2}=-qu^{-1}. \label{G2p}
\end{equation}

\section{Bilinear relation for generic $q_1, q_2$.}
\label{app:toralg}

Bilinear relations on $q$-deformed conformal blocks exist not only in case $q_1=q^{-1}$, $q_2=q$. In order to write the conjecture we introduce bilinear combination 
\begin{equation}
\widehat{\calF}_d(u,q_1,q_2|Z)=\sum_{2n \in \mathbb{Z}}\left(\frac{u^{2dn}(q_1q_2)^{4dn^2}Z^{2n^2}}{(uq_1^{4n-2},u^{-1}q_1^{-4n-2})^{(1)}_{\infty}(u q_1^{-1}q_2^{4n+1},u^{-1}q_1^{-1}q_2^{-4n+1})^{(2)}_{\infty}}\calF_n^{(1)}(q_1^{2d}Z)\, \calF_n^{(2)}(q_2^{2d}Z) \right),
\end{equation}
where we use notations
\begin{align*}
&\calF^{(1)}_n(z)=\calF(u q_1^{4n},q_1^2,q_1^{-1}q_2|z),\qquad \calF^{(2)}_n(z)=\calF(u q_2^{4n},q_1q_2^{-1},q_2^2|z),\\
&(uq_1^{4n-2},u^{-1}q_1^{-4n-2})^{(1)}_{\infty}=(uq_1^{4n-2};q_1^{-2}, q_1^{-1}q_2)_\infty(u^{-1}q_1^{-4n-2};q_1^{-2}, q_1^{-1}q_2)_\infty,\\
&(u q_1^{-1}q_2^{4n+1},u^{-1}q_1^{-1}q_2^{-4n+1})^{(2)}_{\infty}=(u q_1^{-1}q_2^{4n+1};q_1^{-1}q_2,q_2^{2})_\infty(u^{-1}q_1^{-1}q_2^{-4n+1};q_1^{-1}q_2,q_2^2)_\infty,
\end{align*}
and ensure conditions $|q_2|<1<|q_1|$ (in other sectors one can use \eqref{qtrans}).

We will also the following version of Nekrasov partition function
\begin{equation}\label{eq:Flozenge}
\calF_{\lozenge}(u,q_1,q_2|Z)=\calF_{\lozenge,\mathrm{1-loop}}(u;q_1,q_2)\sum_{\lambda_1,\lambda_2| d(\vec{\lambda}=0,1)} Z^{\frac{|\lambda_1|+|\lambda_2|}2}\frac{1}{\prod_{i,j=1}^2 N^{\NS}_{\lambda_i,\lambda_j}({q_1,q_2,u_i}/{u_j})}.
\end{equation}
Here $d(\vec{\lambda})$ denoted the difference between total numbers of white and black boxes in $\lambda_1,\lambda_2$ in chess board coloring. The appropriate Nekrasov functions used in \eqref{eq:Flozenge} have the form
\begin{align}
\!\!\!\!\calF_{\lozenge,\mathrm{1{-}loop}}(u;q_1,q_2)=&\frac{1}{(uq_1^{-2};q_1^{-2},q_2^2)_\infty(uq_1^{-1}q_2;q_1^{-2},q_2^2)_\infty(u^{-1}q_1^{-2};q_1^{-2},q_2^2)_\infty(u^{-1}q_1^{-1}q_2;q_1^{-2},q_2^2)_\infty},
\\
N^{\lozenge}_{\lambda,\mu}(q_1,q_2,u)
=&
\prod_{s\in \lambda^\lozenge}
(1-u q_2^{-a_\mu(s)-1}q_1^{\ell_\lambda(s)}) \cdot 
\prod_{s \in \mu^\lozenge}
(1-u q_2^{a_\lambda(s)}q_1^{-\ell_\mu(s)-1}),
\end{align}
with
$\lambda^{\lozenge}=\{s \in \lambda| a_\mu(s)+\ell_\lambda(s)+1 \equiv0\bmod 2\}$, $\mu^{\lozenge}=\{s \in \mu| a_\lambda(s)+\ell_\mu(s)+1 \equiv0\bmod 2\}$ (cf. \eqref{confblock},\eqref{Nlm}).

\begin{conj}
The following relation holds 
\begin{equation} \label{eq:genbilin}
\widehat{\calF}_0(u,q_1,q_2|Z)=\calF_{\lozenge}(u,q_1,q_2|Z),\qquad \widehat{\calF}_1(u,q_1,q_2|Z)=\left(1-q_2q_1Z^{1/2}\right) \calF_{\lozenge}(u,q_1,q_2|Z).
\end{equation}
\end{conj}
As a combination of two formulas in \eqref{eq:genbilin} one get generalization of Conjecture \ref{qHc} to the case $q_1q_2\neq 1$. 
The formulas \eqref{eq:genbilin} were checked by computer calculation up to $Z^3$ analytically and up to $Z^{12}$ numerically.

Geometrical meaning of the first equality in \eqref{eq:genbilin} is following. Left side and right side are generating functions 
of equivariant Euler characteristics of sheaf $\mathcal{O}$ for Nakajima quiver varieties  for the quiver $A_1^{(1)}$. Another interpretation of the corresponding manifolds are partial compactification of the instantons on the minimal resolution of $\mathbb{C}^2/\mathbb{Z}_2$. But these compactifications are different as $(\mathbb{C}^*)^2$~manifolds, in terms of quiver varieties the stability parameters belong to different chambers. Nevertheless conjectural formula \eqref{eq:genbilin} states that corresponding equiavariant Euler characteristics are the same. In the conformal (or cohomology, or just $q_1,q_2 \rightarrow 1$) limit the corresponding statement was made in \cite[Sec 4.1]{BBFLT} (see also \cite{IMO},\cite{Ohkava}). 

The second equality in \eqref{eq:genbilin} can be viewed as a blow-up equation, similar to one studied in paper \cite{Nakajima:2005}. Its conformal limit coincides with first two formulas in \cite[eq. 4.22]{KGP} (using appropriate version of the AGT relation \cite{BF}, \cite{BMT1}).


\noindent \textsc{Landau Institute for Theoretical Physics, Chernogolovka, Russia,\\
Skolkovo Institute of Science and Technology, Moscow, Russia,\\
National Research University Higher School of Economics, Moscow, Russia,\\
Institute for Information Transmission Problems, Moscow, Russia,\\
Independent University of Moscow, Moscow, Russia}

\emph{E-mail}:\,\,\textbf{mbersht@gmail.com}\\

\noindent\textsc{National Research University Higher School of Economics, Moscow, Russia\\
Skolkovo Institute of Science and Technology, Moscow, Russia,\\
Bogolyubov Institute for Theoretical Physics, Kiev, Ukraine}

\emph{E-mail}:\,\,\textbf{shch145@gmail.com}

\end{document}